\documentclass[conference]{IEEEtran}
\IEEEoverridecommandlockouts
\usepackage{amsmath,amssymb,amsfonts,amsthm}
\usepackage{graphicx}
\usepackage{xcolor}
\usepackage{cite}
\usepackage{booktabs}
\usepackage[hidelinks]{hyperref}

\newtheorem{theorem}{Theorem}
\newtheorem{proposition}{Proposition}
\newtheorem{corollary}{Corollary}
\newcommand{\bv}[1]{\mathbf{#1}}
\newcommand{\bdelta}{\boldsymbol{\delta}}
\newcommand{\Pp}{\bv{P}^{\perp}}

\begin{document}

\title{RSS-Based Verification of Drone Remote ID Position Claims with Unknown Transmit Power}

\author{\IEEEauthorblockN{Cole Dickerson\IEEEauthorrefmark{1},
Ismail Guvenc\IEEEauthorrefmark{1}\IEEEauthorrefmark{2},
Wahab Khawaja\IEEEauthorrefmark{1}, and Ali Cafer Gurbuz\IEEEauthorrefmark{1}}
\IEEEauthorblockA{\IEEEauthorrefmark{1}Department of Electrical and Computer Engineering\\
North Carolina State University, Raleigh, NC 27695, USA}
\IEEEauthorblockA{\IEEEauthorrefmark{2}WolfSky RF Sensing Systems, Raleigh, NC 27695, USA\\
Email: \{jcdicker, iguvenc, wkhawaj, aligurbuz\}@ncsu.edu}
\thanks{This work was supported in part by the NSF under award CNS-2332835 and by the NSF Graduate Research Fellowship Program under DGE-2137100.}}

\maketitle

\begin{abstract}
Broadcast Remote Identification (Remote ID) is a key enabler for integrating unmanned aircraft into the airspace. However, the position it reports is self-declared and can be falsified with low-cost hardware. In this paper, we study how Remote ID receivers can verify the declared position from the received signal strength (RSS) when neither the transmit power nor the path-loss exponent is known. In particular, we derive a Rao score test whose null distribution is exactly chi-squared for any number of packets. Its threshold and detection probability therefore follow in closed form. We then bound the detection probability of any detector that is invariant to the unknown link parameters. This bound saturates as the falsified position recedes along a fixed bearing, so a sufficiently large falsification can remain undetectable. From this limit, we obtain the worst-case bearing in closed form, the observation length needed to close it, and a receiver placement rule. Simulation results show that the proposed test considerably outperforms an estimate-then-threshold detector at low false-alarm rates.
\end{abstract}

\begin{IEEEkeywords}
Remote ID, ASTM F3411, spoofing detection, received signal strength, Rao test, counter-UAS.
\end{IEEEkeywords}

\section{Introduction}

The integration of unmanned aircraft systems (UAS) into the national airspace relies on broadcast Remote Identification (Remote ID), which is now mandatory in the United States~\cite{faa_rid} and the European Union~\cite{eu_945}. Under ASTM F3411~\cite{astm_f3411}, an aircraft broadcasts its identity, position, velocity, and ground control station position at least once per second over Bluetooth or Wi-Fi. Any receiver in range can therefore locate both the aircraft and its operator.

However, the broadcast position is a self-reported global navigation satellite system (GNSS) fix, and nothing in the message set binds it to the physical location of the transmitter. As a result, open-source tools can already broadcast arbitrary Remote ID positions from commodity hardware~\cite{spoofer}. Authentication schemes such as~\cite{drip_arch,arid,brighente} address a different problem. They bind a message to a key, yet an aircraft holding a valid key can sign a false position as easily as a true one. Therefore, verifying a position claim requires an observable that depends on where the transmitter actually is. The received signal strength (RSS) of the broadcast is such an observable, and every Remote ID receiver already measures it.

Physical-layer verification of Remote ID has only recently been studied. In~\cite{ghostbuster}, a log-distance model is fitted to the RSS at multiple receivers, and the resulting ranges are multilaterated. A message is then flagged when the estimated and the claimed positions differ by more than a threshold. Trajectory consistency is checked in~\cite{orion}, RSS multilateration with an interacting multiple model filter localizes a spoofer in~\cite{webb}, and angle of arrival, Doppler shift, and channel gain are combined in~\cite{li_pla}. In our earlier work~\cite{patent}, we introduced a system that detects drones and evaluates their threat by fusing Remote ID with radio frequency (RF) sensing and radar. However, the detectors in~\cite{ghostbuster,webb,li_pla} threshold a distance or a fusion score, so their false-alarm rate must be calibrated by simulation for every deployment. Moreover, none of these works quantifies how large a falsification can remain undetected.

RSS has long been used to detect identity spoofing in wireless networks~\cite{chen_secon}, to verify aircraft positions in automatic dependent surveillance--broadcast (ADS-B)~\cite{strohmeier}, and to perform location verification with unknown transmit power~\cite{yan_commlett}, though prior RSS verifiers assume a known path-loss exponent and an asymptotic null. For stealthy attacks on state estimators, the largest undetectable manipulation is a standard security metric~\cite{urbina}, and the undetectable manipulations form a linear subspace. We show that, for an RSS verifier, they do not form a subspace and can be unbounded in magnitude.

In this paper, we formulate Remote ID position verification as a composite hypothesis test in which the link budget and the path-loss exponent are unknown nuisance parameters. In particular, we derive a Rao score test whose null distribution is exactly chi-squared for any number of packets. Its threshold and detection probability are therefore available in closed form. We then derive an upper bound on the detection probability of any test that is invariant to the nuisance parameters. Moreover, we prove that this bound converges to a finite limit as the falsification grows along a fixed bearing. From this limit, we obtain the worst bearing, the number of packets needed to bound the set of undetectable falsifications, and a receiver placement rule. This limit also quantifies the value of a transmit power field in the Remote ID message. To the best of our knowledge, an exactly calibrated detector and a characterization of the undetectable falsifications have not been reported for Remote ID.

\section{System Model}\label{sec:model}

We consider $K$ Remote ID receivers at known positions $\bv{q}_k\in\mathbb{R}^3$, $k=1,\dots,K$, around a protected asset, as illustrated in Fig.~\ref{fig:geometry}. An unmanned aircraft at true positions $\bv{p}_n$, $n=1,\dots,N$, broadcasts one ASTM F3411 Location/Vector message per epoch, so that at the mandated rate of 1~Hz, $N$ packets span $N$ seconds. All positions are expressed in a local east-north-up frame centered on the asset.

\begin{figure}[t]
\centering
\includegraphics[width=\columnwidth]{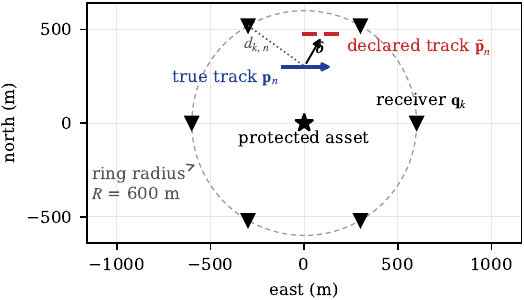}
\caption{Verification geometry with $K=6$ receivers on a ring of radius 600~m. The aircraft flies the solid track and declares the dashed track, displaced by $\bdelta$, while the received power depends on the true ranges $d_{k,n}$.}
\label{fig:geometry}
\end{figure}

\subsection{RSS Model}

The received power in dBm of packet $n$ at receiver $k$ is
\begin{equation}\label{eq:rss}
r_{k,n} = \beta - 10\gamma\log_{10} d_{k,n} + w_{k,n},\qquad d_{k,n}=\|\bv{p}_n-\bv{q}_k\|,
\end{equation}
where $\gamma$ is the path-loss exponent and $w_{k,n}\sim\mathcal{N}(0,\sigma^2)$ is the shadowing in dB, assumed independent across $(k,n)$. The intercept $\beta$ lumps together the equivalent isotropically radiated power (EIRP), the receive antenna gain and the path loss at a 1~m reference distance. This single-slope model with log-normal shadowing is widely used for air-to-ground channels~\cite{khawaja_survey}. A packet is decoded when its mean received power exceeds the receiver sensitivity. We denote the set of decoded pairs $(k,n)$ by $\mathcal{M}$, let $M=|\mathcal{M}|$, and stack the measurements into $\bv{r}\in\mathbb{R}^M$.

The main difficulty is that the ASTM F3411 message set has no transmit power field~\cite{astm_f3411}. The verifier therefore knows neither the EIRP nor the propagation environment, and must treat $\boldsymbol{\theta}=[\beta,\gamma]^{\top}$ as unknown. We refer to this as the \emph{unanchored} case. For comparison, we also consider the \emph{anchored} case, in which $\beta$ is known through, for example, a certified power class, and the fully known case.

\subsection{Falsification Model}

A falsifying aircraft declares the positions
\begin{equation}\label{eq:spoof}
\tilde{\bv{p}}_n = \bv{p}_n + \bdelta,
\end{equation}
where $\bdelta\in\mathbb{R}^2$ is a horizontal offset held constant over the observation window. We restrict $\bdelta$ to the horizontal plane because the message also carries a barometric altitude, and because RSS at ground receivers carries little information about altitude. Further, we assume that the attacker controls the message content but not the emission itself, that is, the transmit power is fixed and the antenna is nominally omnidirectional. In addition, a stealth-aware attacker knows the receiver positions and chooses $\bdelta$ to avoid detection.

Given $\tilde{\bv{p}}_n$ and $\bv{r}$, the verifier tests
\begin{equation}\label{eq:hyp}
\mathcal{H}_0:\ \bdelta=\bv{0}\qquad\text{versus}\qquad \mathcal{H}_1:\ \bdelta\neq\bv{0}.
\end{equation}
To write the likelihood in terms of quantities that the verifier can compute, note from~\eqref{eq:spoof} that the true positions are $\bv{p}_n=\tilde{\bv{p}}_n-\bdelta$. For a hypothesized offset $\bdelta$, we therefore define the vector of log-ranges
\begin{equation}\label{eq:u}
\bv{u}(\bdelta)=\big[\log_{10}\|\tilde{\bv{p}}_n-\bdelta-\bv{q}_k\|\big]_{(k,n)\in\mathcal{M}}\in\mathbb{R}^{M},
\end{equation}
where $\bdelta$ is zero-padded to three dimensions. At the true offset, $\bv{u}(\bdelta)$ equals the true log-ranges. Note that $\bv{u}(\bv{0})$ uses the declared positions directly, without any offset correction. Stacking~\eqref{eq:rss} over $\mathcal{M}$, the mean of $\bv{r}$ is $\beta\bv{1}-10\gamma\,\bv{u}(\bdelta)$, which is linear in $\boldsymbol{\theta}$:
\begin{equation}\label{eq:lin}
\bv{r}\sim\mathcal{N}\big(\bv{H}(\bdelta)\boldsymbol{\theta},\,\sigma^2\bv{I}\big),\qquad \bv{H}(\bdelta)=\big[\bv{1},\,-10\,\bv{u}(\bdelta)\big]\in\mathbb{R}^{M\times 2}.
\end{equation}
Since $\boldsymbol{\theta}$ is unknown under both hypotheses, \eqref{eq:hyp} is a composite test on both sides.

\section{Exactly Calibrated Rao Test}\label{sec:rao}

In this section, we eliminate the link parameters by least squares and explain why existing detectors need calibration by simulation. We then derive a detector whose false-alarm rate depends neither on the geometry, nor on the number of packets, nor on $\boldsymbol{\theta}$.

\subsection{Profiling the Link Parameters}

From~\eqref{eq:lin}, the log-likelihood of $\bv{r}$ for a given $(\bdelta,\boldsymbol{\theta})$ is
\begin{equation}\label{eq:ll}
\ln L(\bdelta,\boldsymbol{\theta})=-\frac{M}{2}\ln(2\pi\sigma^2)-\frac{1}{2\sigma^2}\big\|\bv{r}-\bv{H}(\bdelta)\boldsymbol{\theta}\big\|^2 .
\end{equation}
For a fixed $\bdelta$, maximizing~\eqref{eq:ll} over $\boldsymbol{\theta}$ is a linear least-squares (LS) problem: we regress the measured powers $\bv{r}$ on an intercept and on the hypothesized log-ranges $\bv{u}(\bdelta)$. Its solution is
\begin{equation}\label{eq:ls}
\hat{\boldsymbol{\theta}}(\bdelta)=\big[\hat{\beta}(\bdelta),\hat{\gamma}(\bdelta)\big]^{\top}=\big(\bv{H}^{\top}\bv{H}\big)^{-1}\bv{H}^{\top}\bv{r},
\end{equation}
where $\bv{H}=\bv{H}(\bdelta)$ has full column rank whenever the hypothesized ranges are not all equal. The fitted mean is then $\bv{H}\hat{\boldsymbol{\theta}}(\bdelta)=\bv{P}_{\bv{H}}\bv{r}$, with $\bv{P}_{\bv{H}}=\bv{H}(\bv{H}^{\top}\bv{H})^{-1}\bv{H}^{\top}$ the orthogonal projection onto the column space of $\bv{H}$, and the LS residual is
\begin{equation}\label{eq:resid}
\bv{r}-\bv{H}\hat{\boldsymbol{\theta}}(\bdelta)=\Pp_{\bv{H}}\bv{r},\qquad \Pp_{\bv{H}}=\bv{I}-\bv{P}_{\bv{H}} .
\end{equation}
Substituting~\eqref{eq:ls} and~\eqref{eq:resid} into~\eqref{eq:ll}, the log-likelihood maximized over $\boldsymbol{\theta}$ becomes
\begin{equation}\label{eq:cost}
\ln L\big(\bdelta,\hat{\boldsymbol{\theta}}(\bdelta)\big)=-\frac{M}{2}\ln(2\pi\sigma^2)-\frac{S(\bdelta)}{2\sigma^2},
\end{equation}
where the concentrated cost $S(\bdelta)=\|\Pp_{\bv{H}(\bdelta)}\bv{r}\|^2$ is the residual energy of the best straight-line fit of received power against hypothesized log-range. If the hypothesized positions are correct, $S(\bdelta)$ contains only shadowing; otherwise, the powers no longer lie on a line in $\bv{u}(\bdelta)$ and $S(\bdelta)$ grows.

\subsection{GLRT and Its Calibration Problem}

The generalized likelihood ratio test (GLRT) compares the best fit over a search region $\mathcal{D}$ with the fit at the declared positions. Using~\eqref{eq:cost}, twice the log-ratio of the maximized likelihoods is
\begin{equation}\label{eq:glrt}
T_{\mathrm{G}}=\frac{1}{\sigma^2}\Big[S(\bv{0})-\min_{\bdelta\in\mathcal{D}}S(\bdelta)\Big].
\end{equation}
However, because~\eqref{eq:glrt} involves a nonlinear search over the
region~$\mathcal{D}$, its exact finite-sample null distribution depends
on both~$\mathcal{D}$ and the geometry; Wilks' theorem provides only an
asymptotic $\chi^2_2$ approximation. Its threshold therefore depends on $\mathcal{D}$ and on the geometry, and must be set by simulation for every deployment. The estimate-then-threshold rule of~\cite{ghostbuster}, which compares $\|\hat{\bdelta}\|$ with a threshold for $\hat{\bdelta}=\arg\min_{\bdelta\in\mathcal{D}}S(\bdelta)$, suffers from the same problem.

\subsection{Rao Score Test}

The Rao test avoids the search because it only examines the likelihood in the neighborhood of $\bdelta=\bv{0}$, where all parameters are defined under $\mathcal{H}_0$~\cite{kay_detection}. It needs two ingredients, namely the gradient (score) of the log-likelihood with respect to $\bdelta$ and the Fisher information of $\bdelta$. Both are evaluated at $\bdelta=\bv{0}$ and at the $\mathcal{H}_0$ estimate $\hat{\boldsymbol{\theta}}_0=\hat{\boldsymbol{\theta}}(\bv{0})$ from~\eqref{eq:ls}.

\emph{Score.} The mean $\boldsymbol{\mu}=\beta\bv{1}-10\gamma\bv{u}(\bdelta)$ depends on $\bdelta$ only through $\bv{u}(\bdelta)$. Let
\begin{equation}\label{eq:G}
\bv{G}=\frac{\partial\bv{u}}{\partial\bdelta^{\top}}\bigg|_{\bdelta=\bv{0}}\in\mathbb{R}^{M\times 2},\qquad [\bv{G}]_{m}=-\frac{\bv{e}_{k,n}^{\top}}{\tilde{d}_{k,n}\ln 10},
\end{equation}
where row $m$ corresponds to the pair $(k,n)$, $\tilde{d}_{k,n}=\|\tilde{\bv{p}}_n-\bv{q}_k\|$ is the declared range, and $\bv{e}_{k,n}$ is the horizontal part of the unit vector from $\bv{q}_k$ toward $\tilde{\bv{p}}_n$. Each row of $\bv{G}$ is thus a line-of-sight direction weighted by the inverse range. Differentiating~\eqref{eq:ll} with $\partial\boldsymbol{\mu}/\partial\bdelta^{\top}=-10\gamma\bv{G}$, and noting from~\eqref{eq:resid} that the $\mathcal{H}_0$ residual is $\bv{r}-\bv{H}(\bv{0})\hat{\boldsymbol{\theta}}_0=\Pp\bv{r}$ with $\Pp=\Pp_{\bv{H}(\bv{0})}$, we have
\begin{equation}\label{eq:score}
\bv{s}=\frac{\partial\ln L}{\partial\bdelta}\bigg|_{\bv{0},\hat{\boldsymbol{\theta}}_0}=-\frac{10\hat{\gamma}_0}{\sigma^2}\,\bv{G}^{\top}\Pp\bv{r}.
\end{equation}

\emph{Fisher information.} Since $\boldsymbol{\theta}$ is unknown, the information about $\bdelta$ is reduced by the part that $\boldsymbol{\theta}$ can explain. For the Gaussian model~\eqref{eq:lin}, the Fisher information blocks are $\bv{I}_{\bdelta\bdelta}=(10\gamma/\sigma)^2\bv{G}^{\top}\bv{G}$, $\bv{I}_{\bdelta\boldsymbol{\theta}}=-(10\gamma/\sigma^2)\bv{G}^{\top}\bv{H}(\bv{0})$ and $\bv{I}_{\boldsymbol{\theta}\boldsymbol{\theta}}=\bv{H}(\bv{0})^{\top}\bv{H}(\bv{0})/\sigma^2$. The information that remains for $\bdelta$ after accounting for $\boldsymbol{\theta}$ is their Schur complement,
\begin{equation}\label{eq:J}
\bv{J}=\bv{I}_{\bdelta\bdelta}-\bv{I}_{\bdelta\boldsymbol{\theta}}\bv{I}_{\boldsymbol{\theta}\boldsymbol{\theta}}^{-1}\bv{I}_{\bdelta\boldsymbol{\theta}}^{\top}=\Big(\frac{10\gamma}{\sigma}\Big)^2\bv{G}^{\top}\Pp\bv{G}.
\end{equation}

\emph{Statistic.} The Rao statistic is $T_{\mathrm{R}}=\bv{s}^{\top}\bv{J}^{-1}\bv{s}$ with $\gamma=\hat{\gamma}_0$. Substituting~\eqref{eq:score} and~\eqref{eq:J}, the factor $(10\hat{\gamma}_0)^2$ cancels. Defining $\bv{B}=\Pp\bv{G}$ and using $(\Pp)^{\top}=\Pp=(\Pp)^2$, so that $\bv{G}^{\top}\Pp\bv{G}=\bv{B}^{\top}\bv{B}$, we obtain
\begin{equation}\label{eq:rao}
T_{\mathrm{R}}=\frac{1}{\sigma^2}\bv{r}^{\top}\bv{B}\big(\bv{B}^{\top}\bv{B}\big)^{-1}\bv{B}^{\top}\bv{r}=\frac{1}{\sigma^2}\big\|\bv{Q}^{\top}\bv{r}\big\|^2,
\end{equation}
where $\bv{Q}$ is an orthonormal basis for the column space of $\bv{B}$, so that $\bv{B}(\bv{B}^{\top}\bv{B})^{-1}\bv{B}^{\top}=\bv{Q}\bv{Q}^{\top}$. Two observations follow. First, the test does not require the path-loss exponent, and $\sigma$ enters only as a scale. Second, since the columns of $\bv{Q}$ lie in the range of $\Pp$, they are orthogonal to $\bv{1}$ and $\bv{u}(\bv{0})$. They therefore span the directions that a small falsification changes but the link parameters cannot absorb.

\begin{theorem}\label{thm:exact}
Let $q=\mathrm{rank}(\bv{B})$. Under $\mathcal{H}_0$, $T_{\mathrm{R}}\sim\chi^2_q$ for every $M$ and every geometry. Under $\mathcal{H}_1$, $T_{\mathrm{R}}\sim\chi^2_q(\lambda_{\mathrm{R}})$ with
\begin{equation}\label{eq:lamR}
\lambda_{\mathrm{R}}(\bdelta)=\Big(\frac{10\gamma}{\sigma}\Big)^2\big\|\bv{Q}^{\top}\bv{u}(\bdelta)\big\|^2.
\end{equation}
\end{theorem}
\begin{proof}
Write $\bv{r}=\beta\bv{1}-10\gamma\bv{u}(\bdelta)+\bv{w}$ with $\bv{w}\sim\mathcal{N}(\bv{0},\sigma^2\bv{I})$. Since $\bv{Q}^{\top}\bv{1}=\bv{0}$, we have $\bv{Q}^{\top}\bv{r}=-10\gamma\bv{Q}^{\top}\bv{u}(\bdelta)+\bv{Q}^{\top}\bv{w}$, where $\bv{Q}^{\top}\bv{w}\sim\mathcal{N}(\bv{0},\sigma^2\bv{I}_q)$ because $\bv{Q}^{\top}\bv{Q}=\bv{I}_q$. Under $\mathcal{H}_0$, $\bv{Q}^{\top}\bv{u}(\bv{0})=\bv{0}$, so $T_{\mathrm{R}}=\|\bv{Q}^{\top}\bv{w}\|^2/\sigma^2\sim\chi^2_q$. Under $\mathcal{H}_1$, the mean of $\bv{Q}^{\top}\bv{r}/\sigma$ has squared norm~\eqref{eq:lamR}.
\end{proof}

As a result, for a target false-alarm probability $P_{\mathrm{FA}}$, the threshold and the detection probability are given by
\begin{equation}\label{eq:pd}
\eta=F^{-1}_{\chi^2_q}(1-P_{\mathrm{FA}}),\qquad P_{\mathrm{D}}(\bdelta)=Q_{\chi^2_q(\lambda_{\mathrm{R}}(\bdelta))}(\eta),
\end{equation}
where $F^{-1}_{\chi^2_q}$ is the inverse chi-squared distribution function and $Q_{\chi^2_q(\lambda)}$ is the complementary distribution function of the non-central chi-squared distribution. In practice, $q=2$ whenever $\bv{B}$ has full column rank, which for a hovering aircraft requires at least four receivers. Unlike the GLRT, $T_{\mathrm{R}}$ needs no search, because $\bv{Q}$ depends only on known positions.

\section{Fundamental Limits and Design Rules}\label{sec:limits}

Theorem~\ref{thm:exact} characterizes one particular detector. In this section, we bound the performance of every detector that is invariant to the link parameters. We then show that this bound saturates for large falsifications and derive two design rules.

\subsection{A Bound for Any Invariant Detector}

Since $\beta$ and $\gamma$ are unknown, we restrict attention to
detectors invariant to shifts of the form $\mathbf{r} \mapsto
\mathbf{r} + \mathbf{H}(0)\boldsymbol{\theta}'$ for any
$\boldsymbol{\theta}' \in \mathbb{R}^2$. Such detectors can only
use the component of $\mathbf{r}$ orthogonal to the
nuisance-parameter subspace, namely
\begin{equation}\label{eq:z}
\bv{z}=\bv{U}^{\top}\bv{r}\in\mathbb{R}^{M-2},
\end{equation}
where the columns of $\bv{U}$ form an orthonormal basis of the range of $\Pp$, so that $\bv{U}\bv{U}^{\top}=\Pp$ and $\bv{U}^{\top}\bv{H}(\bv{0})=\bv{0}$. Following the proof of Theorem~\ref{thm:exact}, $\bv{z}\sim\mathcal{N}(\bv{0},\sigma^2\bv{I})$ under $\mathcal{H}_0$ and $\bv{z}\sim\mathcal{N}(\boldsymbol{\mu}_z,\sigma^2\bv{I})$ under $\mathcal{H}_1$, with $\boldsymbol{\mu}_z=-10\gamma\bv{U}^{\top}\bv{u}(\bdelta)$.

\begin{proposition}\label{prop:bound}
For any test of~\eqref{eq:hyp} that is invariant to $\boldsymbol{\theta}$ and has false-alarm probability $P_{\mathrm{FA}}$,
\begin{equation}\label{eq:bound}
P_{\mathrm{D}}(\bdelta)\le Q\big(Q^{-1}(P_{\mathrm{FA}})-\sqrt{\lambda_{\mathrm{c}}(\bdelta)}\big),\quad \lambda_{\mathrm{c}}(\bdelta)=\frac{\|\boldsymbol{\mu}_z\|^2}{\sigma^2},
\end{equation}
where $Q(\cdot)$ is the Gaussian tail function. Moreover, $\lambda_{\mathrm{c}}(\bdelta)=(10\gamma/\sigma)^2\|\Pp\bv{u}(\bdelta)\|^2\ge\lambda_{\mathrm{R}}(\bdelta)$.
\end{proposition}
\begin{proof}
For a fixed $\bdelta$, the best test on $\bv{z}$ is the Neyman--Pearson matched filter $\boldsymbol{\mu}_z^{\top}\bv{z}$. This statistic is Gaussian with variance $\sigma^2\|\boldsymbol{\mu}_z\|^2$ and with mean $0$ under $\mathcal{H}_0$ and $\|\boldsymbol{\mu}_z\|^2$ under $\mathcal{H}_1$, which gives the right-hand side of~\eqref{eq:bound}. A detector that does not know $\bdelta$ cannot do better. The equality follows from $\|\bv{U}^{\top}\bv{u}\|^2=\bv{u}^{\top}\Pp\bv{u}=\|\Pp\bv{u}\|^2$, and the inequality from $\|\bv{Q}^{\top}\bv{u}\|\le\|\bv{U}^{\top}\bv{u}\|$, because the columns of $\bv{Q}$ lie in the range of $\Pp$.
\end{proof}

We refer to $\lambda_{\mathrm{c}}$ as the clairvoyant non-centrality: it is the energy of the RSS signature of the falsification after removing everything the link parameters can absorb. Setting the right-hand side of~\eqref{eq:bound} equal to a target $P_{\mathrm{D}}$ and solving for $\lambda_{\mathrm{c}}$, an operating point $(P_{\mathrm{FA}},P_{\mathrm{D}})$ can be reached only if
\begin{equation}\label{eq:lamstar}
\lambda_{\mathrm{c}}(\bdelta)\ge\lambda^{\star}=\big(Q^{-1}(P_{\mathrm{FA}})-Q^{-1}(P_{\mathrm{D}})\big)^2 .
\end{equation}
We define the \emph{stealth region} $\mathcal{S}=\{\bdelta:\lambda_{\mathrm{c}}(\bdelta)<\lambda^{\star}\}$ as the set of falsifications that no invariant verifier detects at $(P_{\mathrm{FA}},P_{\mathrm{D}})$.

\subsection{Radial Saturation}

Consider now a falsification $\bdelta=\rho\bv{v}$ along a fixed horizontal unit vector $\bv{v}$. Intuitively, one expects $\lambda_{\mathrm{c}}(\rho\bv{v})$ to grow without bound with $\rho$, since a larger falsification should be easier to detect. The following theorem shows that this intuition fails when $\beta$ is unknown.

\begin{theorem}\label{thm:sat}
Let $\boldsymbol{\ell}=[\log_{10}d_{k,n}]_{(k,n)\in\mathcal{M}}$ collect the true log-ranges and let $\bv{a}(\bv{v})$ collect $a_{k,n}=\bv{v}^{\top}(\bv{q}_k-\bv{p}_n)$, using horizontal components. In the unanchored case, provided $a(\mathbf{v})$ is not constant across observations,
\begin{equation}\label{eq:laminf}
\lim_{\rho\to\infty}\lambda_{\mathrm{c}}(\rho\bv{v})=\lambda_{\infty}(\bv{v})=\Big(\frac{10\gamma}{\sigma}\Big)^2\big\|\Pp_{[\bv{1},\,\bv{a}(\bv{v})]}\,\boldsymbol{\ell}\big\|^2 .
\end{equation}
\end{theorem}
\begin{proof}
At $\bdelta=\rho\bv{v}$, the vector $\bv{u}(\rho\bv{v})$ in~\eqref{eq:u} is evaluated at the true positions, so $\bv{u}(\rho\bv{v})=\boldsymbol{\ell}$ for every $\rho$. Hence, only the projector $\Pp$ changes with $\rho$, through the declared ranges. Expanding $\tilde{d}_{k,n}=\|\bv{p}_n+\rho\bv{v}-\bv{q}_k\|$ for large $\rho$ gives $\tilde{d}_{k,n}=\rho-a_{k,n}+O(\rho^{-1})$, and hence
\begin{equation}\label{eq:uexp}
\bv{u}(\bv{0})=\log_{10}\rho\;\bv{1}-\frac{\bv{a}}{\rho\ln 10}+O(\rho^{-2}).
\end{equation}
The first term of~\eqref{eq:uexp} lies in the span of $\bv{1}$, so the column space of $\bv{H}(\bv{0})=[\bv{1},-10\bv{u}(\bv{0})]$ equals that of $[\bv{1},\,\bv{a}+O(\rho^{-1})]$, which converges to $\mathrm{span}\{\bv{1},\bv{a}\}$. Therefore $\Pp\to\Pp_{[\bv{1},\bv{a}]}$, and substituting into $\lambda_{\mathrm{c}}=(10\gamma/\sigma)^2\|\Pp\boldsymbol{\ell}\|^2$ gives~\eqref{eq:laminf}.
\end{proof}

The reason for this saturation is as follows. As the declared position recedes, the declared ranges become large and nearly equal. The free intercept $\beta$ absorbs their common part $\log_{10}\rho$, and the unconstrained exponent $\gamma$ absorbs their first-order variation $\bv{a}$. The only evidence left is therefore the part of the true log-range profile $\boldsymbol{\ell}$ that is not affine in $\bv{a}$. Consequently, if $\lambda_{\infty}(\bv{v})<\lambda^{\star}$, the stealth region is unbounded along $\bv{v}$, and all sufficiently large falsifications in that direction remain undetectable at the operating point.

To find the worst bearing, we evaluate the projection in~\eqref{eq:laminf} in two steps. Projecting out $\bv{1}$ removes the means: let $\tilde{\boldsymbol{\ell}}$ and $\tilde{\bv{a}}$ denote $\boldsymbol{\ell}$ and $\bv{a}$ with their means removed. Projecting out $\tilde{\bv{a}}$ then removes the component of $\tilde{\boldsymbol{\ell}}$ along $\tilde{\bv{a}}$, so that
\begin{equation}\label{eq:twostep}
\big\|\Pp_{[\bv{1},\bv{a}]}\boldsymbol{\ell}\big\|^2=\|\tilde{\boldsymbol{\ell}}\|^2-\frac{(\tilde{\bv{a}}^{\top}\tilde{\boldsymbol{\ell}})^2}{\|\tilde{\bv{a}}\|^2}.
\end{equation}
Since $\bv{a}$ is linear in $\bv{v}$, we can write $\tilde{\bv{a}}=\tilde{\bv{S}}\bv{v}$, where the rows of $\tilde{\bv{S}}$ are the mean-removed horizontal offsets $\bv{q}_k-\bv{p}_n$. With $\bv{b}=\tilde{\bv{S}}^{\top}\tilde{\boldsymbol{\ell}}$ and $\bv{A}=\tilde{\bv{S}}^{\top}\tilde{\bv{S}}$, the last term of~\eqref{eq:twostep} becomes the generalized Rayleigh quotient $(\bv{b}^{\top}\bv{v})^2/(\bv{v}^{\top}\bv{A}\bv{v})$, whose maximum over $\bv{v}$ is $\bv{b}^{\top}\bv{A}^{-1}\bv{b}$. This yields the following result.

\begin{corollary}[Worst bearing]\label{cor:worst}
Provided $\mathbf{A} \succ 0$ and $\mathbf{b} \neq \mathbf{0}$, the bearing that minimizes~\eqref{eq:laminf} is $\bv{v}^{\star}\propto\bv{A}^{-1}\bv{b}$, and
\begin{equation}\label{eq:worst}
\min_{\bv{v}}\lambda_{\infty}(\bv{v})=\Big(\frac{10\gamma}{\sigma}\Big)^2\Big(\|\tilde{\boldsymbol{\ell}}\|^2-\bv{b}^{\top}\bv{A}^{-1}\bv{b}\Big).
\end{equation}
\end{corollary}

Corollary~\ref{cor:worst} removes the need for a bearing search: the security-relevant quantity is a direct computation on the geometry. The same argument also quantifies the value of knowing the transmit power.

\begin{corollary}[Anchored link]\label{cor:anch}
If $\beta$ is known and $\gamma$ is unknown, the limit is $(10\gamma/\sigma)^2\|\tilde{\boldsymbol{\ell}}\|^2$ on every bearing, which upper-bounds~\eqref{eq:worst}. If both parameters are known, $\lambda_{\mathrm{c}}(\rho\bv{v})\to\infty$.
\end{corollary}
\begin{proof}
With $\beta$ known, the verifier subtracts $\beta\bv{1}$ from $\bv{r}$, so only $\bv{u}(\bv{0})$ remains in the regressor. By~\eqref{eq:uexp}, its direction converges to $\bv{1}$, and projecting $\boldsymbol{\ell}$ onto the complement of $\bv{1}$ leaves $\tilde{\boldsymbol{\ell}}$. With both parameters known, nothing is projected out, and $\lambda_{\mathrm{c}}=(10\gamma/\sigma)^2\|\boldsymbol{\ell}-\bv{u}(\bv{0})\|^2$ grows as $\log^2\rho$.
\end{proof}

\subsection{Design Rules}

\emph{Observation length.} For a hovering aircraft, every epoch contributes the same $K$ entries to $\boldsymbol{\ell}$ and $\bv{a}$. Each squared norm in~\eqref{eq:twostep} is therefore $N$ times its single-epoch value, so $\lambda_{\infty}$ grows linearly with the number of packets. The minimum number of packets that bounds the stealth region in every direction is then
\begin{equation}\label{eq:nmin}
N_{\min}=\big\lceil\lambda^{\star}/\lambda^{(1)}_{\infty}\big\rceil,
\end{equation}
where $\lambda^{(1)}_{\infty}$ is~\eqref{eq:worst} evaluated for a single epoch. For a moving aircraft, $N_{\min}$ is the smallest $N$ for which~\eqref{eq:worst} exceeds $\lambda^{\star}$, which bisection finds at negligible cost. Since $\lambda_{\infty}\propto\gamma^2/\sigma^2$, doubling the shadowing standard deviation quadruples $N_{\min}$. Note that $N_{\min}$ is necessary for every invariant verifier, not only for the Rao test.

\emph{Receiver placement.} The limit~\eqref{eq:laminf} depends on the geometry only through the spread of the true log-ranges. Consider receivers on a ring of radius $R$ and a target at ground range $\rho$. When $R\ll\rho$, the target is far from all receivers, so the ranges are nearly equal. When $R\gg\rho$, the target is close to the center of the ring, where all ranges are again nearly equal to $R$. In both cases $\tilde{\boldsymbol{\ell}}\approx\bv{0}$ and $N_{\min}$ is large, so the optimum radius lies in between; Section~\ref{sec:results} shows that it is close to $\rho$.

\section{Numerical Results}\label{sec:results}

In this section, we verify the analysis by Monte Carlo simulation and evaluate the limits of Section~\ref{sec:limits}. Unless otherwise stated, $K=6$ receivers at 3~m height are placed on a ring of radius 600~m. The aircraft flies east at 90~m altitude and 12~m/s, 300~m north of the asset, with 10~dBm EIRP over Bluetooth~5 long range at 2.44~GHz. Each receiver has a sensitivity of $-103$~dBm.

For the channel, we use $\gamma=2.34$ and $\sigma=3.4$~dB, within the range reported in~\cite{khawaja_survey}, so that $N=20$ packets yield $M=120$ decoded measurements. The GLRT searches $\mathcal{D}=[-2,2]^2$~km$^2$ on a 50~m grid, and up to $4\times10^5$ Monte Carlo trials are used. Throughout, the operating point is $(P_{\mathrm{FA}},P_{\mathrm{D}})=(10^{-3},0.9)$, for which $\lambda^{\star}=19.1$ in~\eqref{eq:lamstar}. In comparison, the Rao test with $q=2$ requires $\lambda_{\mathrm{R}}\ge 23.8$ to reach the same point. Since every non-centrality in Section~\ref{sec:limits} scales with $(\gamma/\sigma)^2$, the packet counts reported below scale approximately with $(\sigma/\gamma)^2$. For example, $\sigma=4$~dB instead of 3.4~dB increases them by a factor of 1.38.

\subsection{Calibration and Detection Performance}

We first verify Theorem~\ref{thm:exact}. Under $\mathcal{H}_0$, the empirical mean and variance of $T_{\mathrm{R}}$ are 2.00 and 3.99, against 2 and 4 for $\chi^2_2$, and a Kolmogorov--Smirnov test does not reject the $\chi^2_2$ hypothesis ($p=0.12$). With the analytical threshold for $P_{\mathrm{FA}}=10^{-3}$, the realized false-alarm rate of the Rao test is $0.99\times10^{-3}$. The same threshold yields $1.16\times10^{-3}$ for the GLRT, and this mismatch changes with $\mathcal{D}$ and the geometry.

In Fig.~\ref{fig:roc}, we show the receiver operating characteristics of the Rao test, the GLRT and the estimate-then-threshold detector, together with the bound~\eqref{eq:bound}. Here, the falsification is 200~m at a bearing of $60^{\circ}$, measured counterclockwise from east. The closed-form expression~\eqref{eq:pd} matches the simulated Rao test within 0.005 at every plotted point. In comparison, the GLRT, whose threshold is set from $10^5$ null trials, performs similarly, with $P_{\mathrm{D}}=0.70$ against 0.67 for the Rao test at $P_{\mathrm{FA}}=10^{-3}$. The remaining gap of 0.18 to the bound, which equals 0.85 at the same false-alarm rate, is the price of not knowing the direction of the falsification.

\begin{figure}[t]
\centering
\includegraphics[width=\columnwidth]{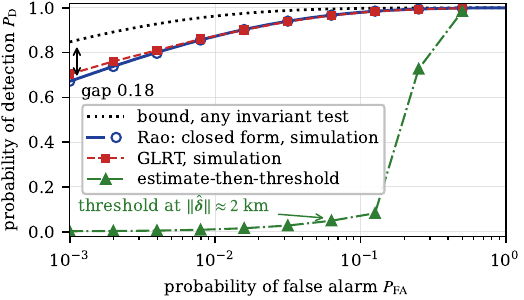}
\caption{Receiver operating characteristics for a 200~m falsification at a bearing of $60^{\circ}$, $K=6$, $M=120$. The bound is~\eqref{eq:bound} with $\lambda_{\mathrm{c}}=16.9$, and the Rao test has $\lambda_{\mathrm{R}}=16.2$. The arrow marks the gap of 0.18 to the bound at $P_{\mathrm{FA}}=10^{-3}$.}
\label{fig:roc}
\end{figure}

On the other hand, the estimate-then-threshold detector attains only $P_{\mathrm{D}}=0.01$ at $P_{\mathrm{FA}}=10^{-2}$, although it uses the same data and the same profiled model. This is because the profiled cost is nearly flat along some directions. As a result, $\hat{\bdelta}$ wanders far from zero under $\mathcal{H}_0$, with a 99th percentile of 2.05~km at the edge of $\mathcal{D}$. The likelihood-based tests avoid this problem, since they measure how much the fit improves rather than how far the estimate moves. Therefore, the Rao test delivers the performance of the GLRT without a search and without calibration by simulation.

\subsection{Radial Saturation}

In Fig.~\ref{fig:sat}, we show $\lambda_{\mathrm{c}}$ as a function of the falsification magnitude on the worst bearing for the three knowledge cases of Section~\ref{sec:model}. By Corollary~\ref{cor:worst}, this bearing is the north--south axis through the aircraft. In the unanchored case, a falsification pointing outward never reaches $\lambda^{\star}$. Instead, $\lambda_{\mathrm{c}}$ increases monotonically to $\lambda_{\infty}=15.9$, so no invariant verifier can detect this falsification at any magnitude. A falsification pointing inward, toward the asset, is easy to detect at a few hundred meters, where $\lambda_{\mathrm{c}}$ peaks at 139. However, it falls below $\lambda^{\star}$ beyond about 7~km and converges to the same limit. As explained in Section~\ref{sec:limits}-B, far away the free intercept and exponent absorb almost all of the signature, so a larger falsification is not necessarily easier to detect.

\begin{figure}[t]
\centering
\includegraphics[width=\columnwidth]{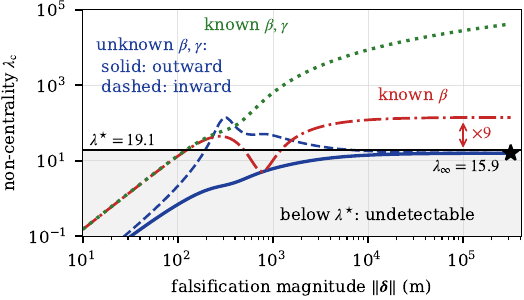}
\caption{Clairvoyant non-centrality against falsification magnitude on the worst bearing, $K=6$, $N=20$, for unknown $\beta,\gamma$ (outward and inward), known $\beta$, and known $\beta,\gamma$. The star marks the closed-form limit~\eqref{eq:laminf}, and the horizontal line is the requirement $\lambda^{\star}$ for $(P_{\mathrm{FA}},P_{\mathrm{D}})=(10^{-3},0.9)$; falsifications in the shaded area cannot be detected. The arrow marks the ninefold gain in the limit when $\beta$ is known.}
\label{fig:sat}
\end{figure}

The curves agree with~\eqref{eq:laminf} and~\eqref{eq:worst} within 0.4\%. For the Rao test, $\lambda_{\mathrm{R}}$ saturates at a lower value of 8.4, as Proposition~\ref{prop:bound} requires. Anchoring $\beta$ raises the limit to 141 on every bearing (Corollary~\ref{cor:anch}), and knowing both parameters removes the saturation altogether. Therefore, in this geometry, a transmit power field in the Remote ID message would increase the worst-case evidence by a factor of nine.

\subsection{Where Verification Fails}

We next evaluate $N_{\min}$ from~\eqref{eq:nmin} over the protected area for a hovering aircraft. Near the center of the ring, a hovering aircraft cannot be verified with any practical window length, because all true ranges are nearly equal there and $\tilde{\boldsymbol{\ell}}\approx\bv{0}$ in~\eqref{eq:worst}. In addition, weak regions appear along the radial lines through each receiver and just outside the ring. Overall, a 20~s window bounds the stealth region for 60\% of the area inside the ring. Moreover, motion helps by adding geometric diversity. Over 400 random positions inside the ring with $K=6$, the stealth region is bounded within 20 packets at 69\% of positions for a transiting aircraft and at 55\% for a hovering one. With $K=3$, a hovering aircraft is never verifiable, since the worst bearing removes the only residual dimension left by the two link parameters.

\subsection{Receiver Placement}

In Fig.~\ref{fig:place}, we show the median $N_{\min}$ over 36 azimuths as a function of the ring radius $R$ for a hovering aircraft at ground ranges $\rho=200$, 400 and 600~m, with $K=8$. Each curve has a sharp minimum close to $R=\rho$. Over $\rho\in[150,600]$~m and $K\in\{6,8,12\}$, the optimal radius lies between $1.0\rho$ and $1.17\rho$, while a ring of half or twice the optimal radius requires 4 to 11 times more packets. As explained in Section~\ref{sec:limits}-C, the spread of the true log-ranges, which carries all the evidence in~\eqref{eq:laminf}, is largest when the target stands among the receivers. Beyond $R=1$~km, distant receivers fall below the sensitivity. Therefore, unlike range-based multilateration, a wider baseline does not help; receivers should be placed at the expected standoff of the threat, not at the asset perimeter.

\begin{figure}[t]
\centering
\includegraphics[width=\columnwidth]{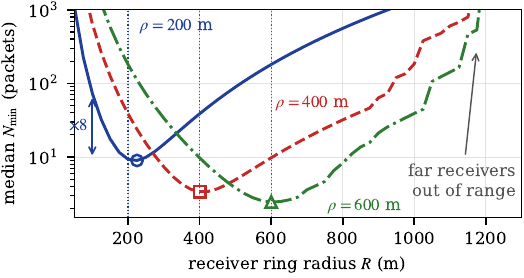}
\caption{Median $N_{\min}$ over target azimuth against ring radius for a hovering aircraft at three ground ranges $\rho$, $K=8$. Open markers show the optima, and vertical dotted lines mark $R=\rho$. The arrow marks the cost of halving the radius for $\rho=200$~m.}
\label{fig:place}
\end{figure}

\section{Conclusion}\label{sec:conc}

In this paper, we studied RSS-based verification of Remote ID position claims with unknown transmit power and path-loss exponent. We derived a Rao score test that is exactly constant false-alarm rate (CFAR) for any number of packets and considerably outperforms an estimate-then-threshold detector at low false-alarm rates. We then proved that the detection probability of every invariant verifier saturates as the falsification grows along a fixed bearing. With 20 packets in our geometry, all sufficiently large outward falsifications are therefore undetectable. The limit also yields the worst bearing, the necessary observation length, and a rule that matches the ring radius to the ground range of the threat. These results assume independent Gaussian shadowing with known variance, calibrated receiver gains, a constant falsification offset and a fixed omnidirectional emission. Future work will address correlated shadowing, which reduces the effective number of packets, power-controlling attackers, and validation on the Aerial Experimentation and Research Platform for Advanced Wireless (AERPAW)~\cite{aerpaw}.

\section*{Acknowledgment}
The authors used Anthropic's Claude Code and Claude Cowork to assist with simulation code development, data analysis, and drafting and editing of the manuscript text. All AI-assisted content was independently verified by the authors; content that could not be verified was revised or discarded. The authors take full responsibility for the work.

\end{document}